\documentclass[11pt, reqno]{amsart}
\usepackage{geometry}
\usepackage{graphicx}
\usepackage{amssymb}
\usepackage{amsmath, mathrsfs, stmaryrd}
\usepackage{comment, tensor}
\usepackage{color}
\newtheorem{theorem}{Theorem}
\newtheorem{definition}[theorem]{Definition}
\newtheorem{lemma}[theorem]{Lemma}

\newtheorem{remark}[theorem]{Remark}

\newcommand{\R}{\mathbb R}
\newcommand{\pl}{\partial}
\renewcommand{\hat}{\widehat}
\newcommand{\na}{\nabla}
\newcommand{\lt}{\left}
\newcommand{\rt}{\right}
\title{Quasi-local mass on unit spheres at spatial infinity}
\author{Po-Ning Chen, Mu-Tao Wang, Ye-Kai Wang, and Shing-Tung Yau}
\numberwithin{theorem}{section}
\numberwithin{equation}{section}
\begin{document}

\begin{abstract}
In this note, we compute the limit of the Wang-Yau quasi-local mass on unit spheres at spatial infinity of an asymptotically flat initial data set. 
Similar to the small sphere limit of the Wang-Yau quasi-local mass, we prove that the leading order term of the quasi-local mass recovers the stress-energy tensor.
For a vacuum spacetime, the quasi-local mass decays faster and the leading order term is related to the Bel-Robinson tensor.  Several new techniques of evaluating quasilocal mass are developed in this note. 
\end{abstract}

\thanks{P.-N. Chen is supported by NSF grant DMS-1308164 and Simons Foundation collaboration grant \#584785, M.-T. Wang is supported by by NSF grant DMS-1405152 and DMS-1810856, Y.-K. Wang is supported by MOST Taiwan grant 105-2115-M-006-016-MY2, 107-2115-M-006-001-MY2, and S.-T. Yau is supported by NSF grants  PHY-0714648 and DMS-1308244. The authors would like to thank the National Center for Theoretical Sciences at National Taiwan University where part of this research was carried out} 

\maketitle
\section{Introduction}
In general relativity, a spacetime is a  4-manifold $N$ with a Lorentzian metric $g_{\alpha\beta}$ satisfying the {\it Einstein equation}
\[R_{\alpha\beta}-\frac{R}{2}g_{\alpha\beta}=8\pi T_{\alpha\beta},\]
where $R_{\alpha\beta}$ and $R$ are the Ricci curvature and the scalar curvature of the metric $g_{\alpha\beta}$, respectively. 
On the right hand side of the Einstein equation, $T_{\alpha\beta}$ is the stress-energy tensor of the matter field, a divergence free and symmetric 2-tensor. For most matter fields, $T_{\alpha\beta}$ satisfies the dominant energy condition. For a vacuum spacetime where $T_{\alpha\beta}=0$ (which implies $R_{\alpha \beta} =0$), one way of measuring the 
gravitational energy  is to consider the {\it Bel-Robinson tensor} \cite{Bel} 
\begin{equation}\label{Q} Q_{\mu\nu \alpha \beta} = W^{\rho \,\,\,\, \sigma}_{\,\,\,\, \mu \,\,\,\, \alpha}W_{\rho \nu \sigma \beta}+W^{\rho \,\,\,\, \sigma}_{\,\,\,\, \mu \,\,\,\, \beta}W_{\rho \nu \sigma \alpha} - \frac{1}{2}g_{\mu\nu}W_{\alpha}^{\,\,\,\,  \rho \sigma \tau}W_{\beta  \rho \sigma \tau},  \end{equation}
where $W_{\alpha \beta \gamma \delta}$ is the Weyl curvature tensor of the spacetime $N$. For a vacuum spacetime, the Bel-Robinson tensor is a divergence free and totally symmetric 4-tensor which also satisfies a certain positivity condition \cite[Lemma 7.1.1]{Christodoulou-Klainerman}.

We recall that given a spacelike 2-surface $\Sigma$ in a spacetime $N$, the Wang-Yau quasi-local energy $E(\Sigma, \mathcal X, T_0)$ (see \eqref{qle}) is defined in \cite{Wang-Yau1,Wang-Yau2} with respect to each pair $(\mathcal X, T_0)$ of an isometric embedding $\mathcal X$ of $\Sigma$ into the Minkowski space $\R^{3,1}$
and a constant future timelike unit vector $T_0\in \R^{3,1}$. If the spacetime satisfies the dominant energy condition and the pair $(\mathcal X,T_0)$ is admissible (see \cite[Definition 5.1]{Wang-Yau2}), it is proved that $E(\Sigma, \mathcal X, T_0) \ge 0$. 
The Wang-Yau quasi-local mass is defined to be the infimum of the quasi-local energy among all admissible pairs $(\mathcal X,T_0)$. The Euler-Lagrange equation for the critical points of the quasi-local energy is derived in  \cite{Wang-Yau2}. The Euler-Lagrange equation, coupled with the isometric embedding equation, is referred to as the optimal embedding equation, see \eqref{oee}. A solution to the equation is referred to as an optimal embedding.

When studying different notions of quasi-local energy, it is natural to evaluate the large sphere and the small sphere limits of the quasi-local energy and compare with the known measures of the gravitational energy in these situations. One expects the following \cite{Christodoulou-Yau, Penrose}:

\

1) For a family of surfaces approaching the spatial/null infinity of an isolated system (the large sphere limit), the limit of the quasi-local energy recovers the total energy-momentum of the isolated system.

2) For a family of surfaces approaching a point $p$ (the small sphere limit), the limit of the quasi-local energy recovers the stress-energy tensor for spacetimes with matter fields and the Bel-Robinson tensor for vacuum spacetimes.

\

There are many works on evaluating the large sphere and the small sphere limits of different notions of quasi-local energy. See for example \cite{blk1,blk2,Chen-Wang-Yau1,Chen-Wang-Yau2,fst,Geroch,Horowitz-Schmidt,Huiken-Ilmanen,Kwong-Tam,Miao-Tam-Xie,Wang-Yau3,Wiygul,yu}.  The list here is by no means exhaustive. For a more comprehensive review of different notions of quasi-local energy and their limiting behaviors, see \cite{sz}  and the references therein.

In a series of papers \cite{Chen-Wang-Yau1,Chen-Wang-Yau2,Wang-Yau3}, the above expectations for the Wang-Yau quasi-local energy were verified. One of the key observations in \cite{Wang-Yau3} (see \cite[Theorem 2.1]{Wang-Yau3}) is that for a family of surfaces $\Sigma_r$ and isometric embeddings $\mathcal X_r$, the limit of $E(\Sigma_r, \mathcal X_r, T_0)$ is a linear function of $T_0$ under the compatibility condition
\begin{equation} \label{ratio}
\lim_{r\rightarrow r_0} \frac{|H_0|}{|H|}=1,
\end{equation}
where $H$ and $ H_0$ are the mean curvature vectors of $\Sigma_r$ in $N$ and  the image of the isometric embedding $\mathcal X_r$ in $\R^{3,1}$, respectively. The compatibility condition \eqref{ratio} holds naturally in the large sphere limit ($r_0=\infty$) at both spatial and null infinity and the small sphere limit ($r_0=0$) around a point. In particular, \cite[Theorem 2.1]{Wang-Yau3} is used throughout the sequence of papers \cite{Chen-Wang-Yau1,Chen-Wang-Yau2,Wang-Yau3}.

In addition to the large sphere limit and the small sphere limit, there is another interesting situation where the compatibility condition holds naturally, namely, the limit of the quasi-local mass on unit spheres at infinity of an asymptotically flat spacetime. In a series of papers \cite{Chen-Wang-Yau6,Chen-Wang-Yau9}, we evaluated the limit at null infinity to capture the information of gravitational radiation. In particular, this is carried out in \cite{Chen-Wang-Yau9} for the Vaidya spacetime. In this note, we evaluate the limit for unit spheres at spatial infinity of an asymptotically flat spacetime, namely, at infinity of an asymptotically flat initial data set.

\begin{theorem} \label{main_theorem}
Let $(M, g, k)$ be an asymptotically flat initial data set as in \eqref{initial_data}. Let $\gamma$ be a geodesic on $M$ which is parametrized by arc-length and extends to infinity. Let $p=\gamma(d)$ be a point on $\gamma$ and $\Sigma$ be the unit geodesic sphere in $M$ that is centered at $p=\gamma(d)$.   The quasi-local mass $E(\Sigma, \mathcal X, T_0)$ for  $T_0 = (a_0, -a_1, -a_2, -a_3)$ has the following asymptotic behavior as $d\rightarrow \infty $ for each of the following isometric embeddings $\mathcal X$ of $\Sigma$.
\begin{enumerate}
\item For the isometric embedding $\mathcal X:\Sigma \rightarrow \R^3$, we have
\begin{align*}
E(\Sigma, \mathcal X, T_0) = \frac{1}{6} \lt( a_0 \mu(p) - a_i J^i(p) \rt) + O(d^{-3-2\alpha}),
\end{align*} where $\mu$ and $J^i$ are defined in \eqref{Hamiltonian constraint} and \eqref{momentum constraint}. 
\item Suppose the initial data set $(M,g,k)$ satisfies the vacuum constraint equation \eqref{vc}.  Let $N$ be the future development of $(M, g, k)$ with Weyl curvature $\bar W$. Let  $e_0, e_1, e_2, e_3$ be an orthonormal basis at $p$ with $e_0$ the unit timelike normal of $M$ in $N$. For $(\mathcal X,T_0)$ solving the leading order of the optimal embedding equation \eqref{oee}, we have
\begin{align*}
E(\Sigma, \mathcal X, T_0) &= \frac{1}{90} \lt( Q(e_0,e_0,e_0,T_0) + \frac{1}{2a_0} \bar W_{0i0j} \bar W\indices{_0^i_0^j} \rt) + O(d^{-4-3\alpha}). 
\end{align*}
\end{enumerate}
Here $\bar W_{0i0j} = \bar W(e_0, e_i, e_0 ,e_j)(p)$, $Q$ is the Bel-Robinson tensor of $N$ at $p$, and $T_0$ is identified with the timelike vector $a_0 e_0 + \sum_{i=1}^3 a_i e_i$ at $p$.
\end{theorem}

Our investigation begins with the Brown-York mass. We compute the derivative of the Brown-York mass and use it to rewrite the Brown-York mass as a bulk integral. The integrand consists of the scalar curvature and quadratic terms of the difference of the physical and the reference data, see Lemma \ref{lemma_Brown_York}. We use Lemma \ref{lemma_Brown_York}  to evaluate the limit in Theorem \ref{Brown_York_thm}. The scalar curvature corresponds to the stress-energy tensor whereas the quadratic terms, which decay faster than the scalar curvature, correspond to the Bel-Robinson tensor. 

In the remaining part of this article, we consider the Wang-Yau quasi-local mass for initial data sets which are not necessarily time-symmetric. We start by solving the optimal embedding equation. The structure of the equation is similar to that of \cite{Chen-Wang-Yau2} for the small sphere limit and of \cite{Chen-Wang-Yau9} for unit spheres at null infinity of the Vaidya spacetime. After obtaining the optimal embedding, we use it as the Dirichlet boundary value to solve Jang's equation in the bulk. Using the Schoen-Yau identity from \cite{Schoen-Yau} and the canonical gauge for the quasi-local mass from \cite{Wang-Yau1}, we obtain Theorem \ref{mass by Jang's equation} which generalizes Theorem \ref{Brown_York_thm} for the Wang-Yau quasi-local mass. While the formula is more complicated, it still consists of the integral of the stress-energy tensor and some quadratic terms. In Section 5, we compute the terms appearing in Theorem \ref{mass by Jang's equation} explicitly and evaluate the limit. Theorem \ref{main_theorem} is obtained after assembling these results. We observe that the answer is very similar to the small sphere limit obtained in \cite{Chen-Wang-Yau2}. In Section 6, we demonstrate how the new approach of this article can be applied to recover the result of \cite{Chen-Wang-Yau2} for the small sphere limit.

\section{Review of the Wang-Yau quasi-local mass and asymptotical flatness}
Let $\Sigma$ be a closed spacelike 2-surface in a spacetime $N$ with spacelike mean curvature vector $H$. Denote the induced metric and connection one-form of $\Sigma$ by $\sigma$ and
\[ \alpha_H(\cdot) = \lt\langle \na^N_{(\cdot)} \frac{J}{|H|}, \frac{H}{|H|} \rt\rangle\]
where $J$ is the reflection of $H$ through the incoming light cone in the normal bundle. Given an isometric embedding $\mathcal X:\Sigma \rightarrow \R^{3,1}$ and future timelike unit Killing field $T_0$ in $\R^{3,1}$, we consider the projected embedding $\hat{\mathcal X}$ into the orthogonal complement of $T_0$, and denote the induced metric and the mean curvature of the image surface $\hat\Sigma$ by $\hat\sigma$ and $\hat H$.

The quasi-local energy with respect to $(\mathcal X,T_0)$ is
\begin{equation}\label{qle} E(\Sigma,\mathcal X,T_0) = \frac{1}{8\pi} \int_{\hat\Sigma} \hat H d\hat\Sigma - \frac{1}{8\pi} \int_{\Sigma} \lt( \sqrt{1 + |\na\tau|^2} \cosh\theta |H| - \na\tau \cdot \na\theta - \alpha_H(\na\tau )\rt) d\Sigma,\end{equation}
where $\na$ and $\Delta$ are the gradient and Laplace operator of $\sigma$, $\tau = - \langle \mathcal X, T_0 \rangle$ is considered as a function on the 2-surface, and
\[ \theta = \sinh^{-1} \lt( \frac{-\Delta\tau}{|H| \sqrt{1+|\na\tau|^2}} \rt). \]
Moreover, we say that $\tau$ solves the {\it optimal embedding equation} if
\begin{equation}\label{oee} div_{\sigma} \lt( \rho \na\tau - \na \lt[ \sinh^{-1} \lt( \frac{\rho\Delta\tau}{|H_0||H|}\rt)\rt] - \alpha_{H_0} + \alpha_H \rt) =0,\end{equation}
where
\[ \rho = \frac{\sqrt{|H_0|^2 + \frac{(\Delta\tau)^2}{1+|\na\tau|^2}} - \sqrt{|H|^2 + \frac{(\Delta\tau)^2}{1+|\na\tau|^2}}}{\sqrt{1+|\na\tau|^2}}. \]

Next we recall the definition of an asymptotically flat initial data set. 
\begin{definition}\label{initial_data}
$(M^3,g,k)$ is an {\it asymptotically flat} initial data set if, outside a compact set,
$M^3$ is diffeomorphic to  $ \R^3  \setminus \{ | x | \le r_0 \}$ for some $r_0>0$
and under the diffeomorphism, we have
\begin{equation}\label{asym_flat_g}
g_{ij} - \delta_{ij} = O ( | x |^{- \alpha} ) ,\  \partial g_{ij} = O ( |x|^{-1-\alpha}), \  \partial^2 g_{ij} = O (|x|^{-2 - \alpha} ), \  \partial^3 g_{ij} = O (|x|^{-3 - \alpha} ), 
\end{equation}
and
\begin{equation}\label{asym_flat_k}
k_{ij}  = O ( | x |^{- 1-\alpha} ) ,  \partial k_{ij} = O ( |x|^{-2-\alpha}), \  \partial^2 k_{ij} = O (|x|^{-3 - \alpha} )
\end{equation}
for some $\alpha>\frac {1}{2}$. Here $ \partial $ denotes the partial differentiation on $ \R^3$. Furthermore, we shall assume that for the constraint equation, we have
\begin{align}
\frac{1}{2} \lt( R(g) + (tr k)^2 - |k|^2 \rt)= \mu, \quad \mu  &= O (|x|^{-3 - \alpha} ), \pl \mu = O(|x|^{-4-\alpha})\label{Hamiltonian constraint} \\
D^i (k_{ij} - (trk) g_{ij}) = J_j, \quad J &=O (|x|^{-3 - \alpha} ), \pl J = O(|x|^{-4-\alpha})  \label{momentum constraint}
\end{align}
\end{definition}
Recall that an initial data set satisfies the dominant energy condition if
\[
\mu \ge |J|.
\]
On the other hand, an initial data set satisfies the vacuum constraint equation if 
\begin{equation}\label{vc}
\mu = 0  \   {\rm and}  \ J =0.
\end{equation}
In this case, there is a unique spacetime $N$ with initial data $(M^3,g,k)$ which solves the vacuum Einstein equation.

Let $\gamma$ be a geodesic on $M$ which is parametrized by arc-length and extends to infinity. We consider $p=\gamma(d)$ for $d\rightarrow \infty$. 
Consider the normal coordinate $(X^1,X^2,X^3)$ centered at $p$ and let $\Sigma(s)$ be the sphere of radius $s$ in the normal coordinate. 
The goal is to evaluate the quasi-local mass of the surface $\Sigma=\Sigma(1)$. In particular, we are interested in the leading order term in $d$.

The set-up of our calculation can be described as the following. On a unit ball $B$ of $\R^3$, there is a family of Riemannian metrics $g_{ij}(d)$ and symmetric 2-tensors $k_{ij}(d)$ parametrized by $d$. The metrics $g_{ij}(d)$ (the symmetric 2-tensors $k_{ij}(d)$, respectively) are the pull back of the metrics (the symmetric 2-tensors, respectively) on the unit geodesic ball centered at $\gamma (d)$, $d_0\leq d<\infty$, a geodesic on $M$ that extends to spatial infinity. We assume that
\begin{enumerate}
\item the standard Cartesian coordinate system $(X^1, X^2, X^3)$ is a geodesic coordinates system for each $g_{ij}(d)$ such that the origin of the coordinates system corresponds to $\gamma(d)$;
\item with respect to $(X^1, X^2, X^3)$, the asymptotic flat conditions \eqref{asym_flat_g} and \eqref{asym_flat_k} are satisfied with $|x|$ replaced by $d$;
\item the constraint equations \eqref{Hamiltonian constraint} and \eqref{momentum constraint} are satisfied with $|x|$ replaced by $d$. 
\end{enumerate}
Namely,
\begin{equation}
g_{ij}(d) - \delta_{ij} = O ( d^{- \alpha} ) ,\  \partial g_{ij}(d) = O ( d^{-1-\alpha}), \  \partial^2 g_{ij}(d) = O (d^{-2 - \alpha} ), \  \partial^3 g_{ij}(d) = O (d^{-3 - \alpha} ),
\end{equation}
and
\begin{equation}
k_{ij}(d)  = O ( d^{- 1-\alpha} ) ,  \partial k_{ij}(d) = O ( d^{-2-\alpha}), \  \partial^2 k_{ij}(d) = O (d^{-3 - \alpha} )
\end{equation}
for some $\alpha>\frac {1}{2}$, where $ \partial $ now denotes the partial differentiation with respect to $X^1,X^2,X^3$. Moreover,
\begin{align}
\frac{1}{2} \lt( R(g) + (tr k)^2 - |k|^2 \rt)= \mu, \quad \mu  &= O (d^{-3 - \alpha} ), \pl \mu = O(d^{-4-\alpha})\\
D^i (k_{ij} - (tr k) g_{ij}) = J_j, \quad J &=O (d^{-3 - \alpha} ), \pl J = O(d^{-4-\alpha}) .
\end{align}

In particular, let $R_{ij}(d)$ denote the Ricci curvature of $g_{ij}(d)$, by the Taylor expansion at a point in $B$ with respect to the geodesic coordinate system $(X^1, X^2, X^3)$, we have 
\begin{align}\label{expansion_geodesic}
g_{ij}(d)(X^1, X^2, X^3) &= \delta_{ij}-\frac{1}{3}R_{ikjl}(d) (0, 0, 0) X^k X^l + O(d^{-3-\alpha}),\\
R_{ij}(d)(X^1, X^2, X^3) &= R_{ij}(d)(0, 0, 0) + O(d^{-3-\alpha}),\\
k_{ij}(d)(X^1, X^2, X^3) &= k_{ij}(d)(0, 0, 0) + \pl_m k_{ij}(d) (0, 0, 0) X^m + O(d^{-3-\alpha}).
\end{align}
These expansions will be abbreviated as 
\begin{align}\label{abbreviated expansion}
g_{ij}&= \delta_{ij}-\frac{1}{3}R_{ikjl}(p) X^k X^l + O(d^{-3-\alpha}),\\
R_{ij} &= R_{ij}(p) + O(d^{-3-\alpha}),\\
k_{ij}&= k_{ij}(p) + \pl_m k_{ij}(p) X^m + O(d^{-3-\alpha}).
\end{align}
We will also use the spherical coordinate system $(s, u^1, u^2)$ on $B$ such that the coordinate transformation $(s, u^1, u^2)\mapsto (X^1, X^2, X^3)$ is given by  $X^i=s \tilde{x}^i (u^1, u^2), i=1, 2, 3$, where $\tilde{x}^i, i=1, 2, 3$ are the three standard coordinate functions on the standard unit sphere in $\R^3$. \\

\noindent{\it Notation:} Einstein summation notation will be used throughout the paper, where $i,j,\cdots$ sum from 1 to 3. Since we are working in normal coordinates, we can freely raise or lower indices for tensors at $p=\gamma(d)$. 
\section{The Brown-York mass}
In this section, we consider a time-symmetric initial data set and compute the limit of the Brown-York quasi-local mass. The starting point is the following lemma for the Brown-York quasi-local mass: Given a surface $\Sigma$ in a 3-manifold $(M,g)$, let  $R$ be the scalar curvature of $g$. Let $\Omega$ be the region in $M$ bounded by $\Sigma$. Suppose $\Omega$ is foliated by surfaces $\Sigma(s)$  with positive Gauss curvature where $0 < s \le 1$,  
$ \Sigma(1) = \Sigma $, and $\Sigma(s)$ shrinks to a point as $s$ tends to $0$. Let $\sigma(s)$ be the induced metric on $\Sigma(s)$. The positivity of the Gauss curvature of $\sigma(s)$ guarantees an isometric embedding into $\R^3$. Denote the mean curvature of  $\Sigma(s)$ in $M$ by $H(s)$ and the mean curvature of the isometric embedding of $\Sigma(s)$  into $\R^3$ by  $H_0(s)$. Let $h(s)$ and $h_0(s)$ be the second fundamental form of $\Sigma(s)$ in $M$ and $\R^3$, respectively.
\begin{lemma} \label{lemma_Brown_York}
The Brown-York quasi-local mass, $m_{BY}(\Sigma)$, of $\Sigma$ is
\[
m_{BY}(\Sigma) = \frac{1}{ 16 \pi}\int_{\Omega } \left( |h_0(s)-h(s)|^2 -   (H_0(s) - H(s)) ^2 + R \right)
\]
where $R$ is the scalar curvature of $g$. 
\end{lemma}
\begin{proof}
Assume that $\Sigma(s)$ are given by $F(x,s): \Sigma \times (0,1] \rightarrow M$ with $DF_{(x,s)}(\frac{\pl}{\pl s}) = f(x,s) \nu(x,s)$ where $\nu(x,s)$ is the unit normal of $\Sigma(s)$. We first show that the derivative of the Brown York quasi-local mass is given by 
\begin{align}\label{variation_Brown_York}
\frac{{\rm d}}{ {\rm d } s} E(\Sigma(s)) = \int_{\Sigma(s)} \frac{f}{2} \left( |h_0-h|^2 -   (H_0 - H) ^2 + R \right). 
\end{align}
The above formula is known, see \cite[Theorem 3.1]{Miao-Shi-Tam} for example. For completeness, we include the proof here. We have
\[
\frac{{\rm d}}{ {\rm d } s} \sigma(s) = 2f h.
\]
By Proposition 6.1 of \cite{Wang-Yau2}, 
\begin{equation} \label{Brown_York_1}
\frac{{\rm d}}{ {\rm d } s} \int_{\Sigma(s)} H_0(s)   = \int_{  \Sigma(s)}  f (H_0 H - h \cdot h_0 ).
\end{equation}
On the other hand, from the second variation formula, we have
\begin{equation} \label{Brown_York_2}
\frac{{\rm d}}{ {\rm d } s} \int_{\Sigma(s)} H(s) = \int_{  \Sigma(s)} f (H^2 -Ric(\nu,\nu) - |h|^2 ).
\end{equation}
The Gauss equations of  $\Sigma(s)$ in $\R^3$ and $M$ imply:
\[
\begin{split}
K = & \frac{1}{2} (H_0^2 - |h_0|^2),  \\
K = & \frac{R}{2}  -Ric(\nu,\nu) + \frac{1}{2} (H^2 - |h|^2)  .
\end{split}
\]
Taking the difference of the two Gauss equations, we obtain
\begin{equation} \label{Brown_York_3}
Ric(\nu,\nu)=  \frac{R}{2} +  \frac{1}{2} (H^2 - |h|^2 - H_0^2 + |h_0|^2) .
\end{equation}
The claim follows from subtracting \eqref{Brown_York_2} from  \eqref{Brown_York_1} and using  \eqref{Brown_York_3} to replace the Ricci curvature term in the result. The lemma follows from integrating \eqref{variation_Brown_York} along the foliation.
\end{proof}

In our setup, $\Sigma(s)$ is the sphere of radius $s$ in the normal coordinates centered at $p=\gamma(d) \in M$. The induced metric and second fundamental form of $\Sigma(s)$ are given by
\begin{align*}
\sigma_{ab} &= s^2 ( \tilde \sigma_{ab} - \frac{1}{3} R_{ikjl}(p) \tilde x^i_a \tilde x^j_b \tilde x^k \tilde x^l ) + O(d^{-3-\alpha})
\\
h(s)_{ab} &= s \tilde\sigma_{ab} + O(d^{-2-\alpha})\end{align*}
where $R_{ikjl} $ is the Riemann curvature tensor of the metric $g$ on $M$ and $\tilde x^i_a$ is a shorthand for $\pl_a \tilde x^i$.

We first compute the difference of the mean curvature and second fundamental form of $\Sigma(s)$ in $M$ and $\R^3$.
\begin{lemma} \label{lemma_mean_curvature_diff_1}
Consider the surface $\Sigma(s)$. We have
\begin{equation}
H_0(s)- H(s) =  -s R_{ij}(p)\tilde x^i \tilde x^j + O(d^{-3 - \alpha}) 
\end{equation}
and
\begin{equation}
h_0(s) - h(s) = s^3 R_{ij}(p) \tilde x^i_a \tilde x^j_b + O(d^{-3 - \alpha}).
\end{equation} 
\end{lemma}
\begin{proof}
We will use repeatedly an implication of \eqref{Hamiltonian constraint}  that $R = O(d^{-3-\alpha})$. Let $\hat h$ and $\hat h_0$ denote the traceless second fundamental forms. By the Gauss equations, 
\begin{align*}
2K &= R - 2 Ric(\nu,\nu) + \frac{1}{2} H^2 - \frac{1}{2} |\hat h|^2\\
&= \frac{1}{2} H_0^2 - \frac{1}{2} |\hat{h_0}|^2.
\end{align*}
Since the unit normal of $\Sigma(s)$ is $\nu = \tilde x^i \frac{\pl}{\pl X^i} + O(d^{-3-\alpha})$ and $Ric(\nu,\nu) = R_{ij}\tilde x^i \tilde x^j + O(d^{-3-\alpha})$, we get $H_0(s) - H(s) = -s R_{ij} \tilde x^i \tilde x^j + O(d^{-3-\alpha})$.

Taking the difference of the Codazzi equations for $\Sigma(s)$ in $M$ and $\R^3$ implies
\begin{align*}
\tilde \na^a (\hat h_{ab}(s) - \hat h_{0ab} (s)) &= -\frac{1}{2} \pl_b (H_0(s) - H(s)) - Ric(\nu,\pl_b) + O (d^{-4-2\alpha})\\
&= \pl_b (s R_{ij} \tilde x^i \tilde x^j) + O(d^{-3-\alpha}).
\end{align*}
One readily checks that
\[ \hat h_{ab} (s) - \hat h_{0ab}(s) = s^3 R_{ij} \tilde x^i_a \tilde x^j_b + \frac{s^3}{2} R_{ij} \tilde x^i \tilde x^j \tilde\sigma_{ab} \]
satisfies the above equation. Indeed, we find the unique solution as there is no divergence-free, traceless symmetric 2-tensor on $S^2$.
\end{proof}
We obtain the following result for the limit of the Brown-York mass:
\begin{theorem} \label{Brown_York_thm}
On $\Sigma=\Sigma(1)$, we have
\[
m_{BY}(\Sigma) = \frac{1}{6} \mu(p)+O(d^{-4 -\alpha}).
\]

For an initial data set satisfying the vacuum constraint equation, we have
\[
m_{BY}(\Sigma) = \frac{1}{ 60} Q(e_0,e_0,e_0,e_0)+O(d^{-5 -2 \alpha})
\]
where $Q$ is the Bel-Robinson tensor \eqref{Q} at $p$ of the solution to the vacuum Einstein equation with the time-symmetric initial data $(M,g)$.
\end{theorem}
\begin{proof}
On $\Omega$, $R = R(p) + O(d^{-4-\alpha})$ by (\ref{Hamiltonian constraint}). Applying Lemma \ref{lemma_Brown_York}, together with  
\begin{equation} 
\begin{split}
|h(s)-h_0(s)| = &O(d^{-2 - \alpha})\\
H_0(s)- H(s) =  &O(d^{-2-\alpha}),
\end{split}
\end{equation} 
we obtain the first formula.

For a vacuum initial data set, we compute
\begin{align*}
|h_0(s) - h(s)|^2 - (H_0(s) - H(s))^2 = s^2 R_{ij} R_{lm} (\delta^{il}\delta^{jm} - \delta^{il} \tilde x^j \tilde x^m - \delta^{jm} \tilde x^i \tilde x^l) + O(d^{-5-2\alpha}).
\end{align*}
Since $R=0$, we have
\begin{align*}
m_{BY}(\Sigma)= \frac{1}{16\pi} R_{ij} R^{ij} \cdot \frac{4\pi}{3} \int_0^1 s^4 ds +O(d^{-5-2\alpha})= \frac{1}{60} R_{ij}R^{ij} +O(d^{-5-2\alpha}).
\end{align*}
Finally, if $N$ is the solution to the Einstein equation with time-symmetric initial data $(M,g)$, its Weyl curvature satisfies
\begin{align*}
\bar W_{0i0j}&= R_{ij}, \\
\bar W_{0ijk} &=0.
\end{align*}
This finishes the proof of the theorem.
\end{proof}
\section{Optimal embedding equation and the Jang equation}
In this section, we describe our strategy to handle the second fundamental form  $k_{ij}$. We study the optimal embedding equation on $\Sigma$ and the Jang equation on $\Omega$. In particular, we first solve the optimal embedding equation on the boundary. Then we solve the Jang equation on the bulk $\Omega$ using the solution of the optimal embedding equation as the boundary value

Consider the product manifold $\Omega \times \R$ with the product metric $dt^2 + g_{ij} dX^i dX^j$. The data $k_{ij}, \mu, J$ are extended parallelly along the $\R$ factor. Jang's equation for $u \in C^2(\Omega)$ reads
\begin{align}
\lt( g^{ij} - \frac{D^iu D^ju}{1 + |D u|^2} \rt) \lt( k_{ij}-\frac{D_i D_j u}{\sqrt{1 + |Du|^2}} \rt)=0.
\end{align}  

Denote the graph of $u$ in $\Omega\times \R$ by $\tilde\Omega$ and $\tilde\Sigma = \pl \tilde\Omega$. Let $\tilde e_4$ be the downward normal of $\tilde\Omega$ and $Y_i = \lt( k - \frac{D^2 u}{\sqrt{1+|Du|^2}} \rt)\lt( \frac{\pl}{\pl X^i}, \tilde e_4 \rt) $.

Let $\tilde g$ be the induced metric of $\tilde\Omega$. Let $H_0$ be the mean curvature of the isometric embedding of $\tilde\Sigma$ into $\R^3$. We recall
\[ E(\Sigma, \mathcal X, T_0) = \frac{1}{8\pi} \lt( \int H_0 d\tilde\Sigma - \int \lt[ \sqrt{1+|\na\tau|^2} \cosh |H| + \Delta \tau \theta - \alpha_H(\na\tau) \rt] d\Sigma\rt)\]
where $\sinh\theta = \frac{-\Delta\tau}{|H|\sqrt{1+|\na\tau|^2}}$. By \cite[Theorem 4.1]{Wang-Yau2},
\[ 
\begin{split}
 & E(\Sigma, \mathcal X, T_0) \\
 = &\frac{1}{8\pi} \lt( \int \lt[ H_0 - \tilde H + \langle Y, \tilde e_3 \rangle \rt] d\tilde\Sigma + \int \lt[ |H|\sqrt{1+|\na\tau|^2} (\cosh\theta' - \cosh\theta) + \Delta\tau (\theta' - \theta)\rt] d\Sigma \rt),
 \end{split}
 \]
where $\theta'$ is defined by
\begin{align*}
\cosh\phi e_3 - \sinh\phi e_4 = \cosh\theta' e_3^H - \sinh\theta' e_4^H
\end{align*}
with $\sinh\phi = - \frac{u_3}{\sqrt{1+|\na\tau|^2}}$.
By the Schoen-Yau identity \cite[(2.29)]{Schoen-Yau}
\[ 2(\mu - J(\tilde{e}_4) ) = \tilde R - \lt| k_{ij} - \frac{D_i D_j u}{\sqrt{1+|Du|^2}} \rt|^2_{\tilde g} - 2 |Y|^2_{\tilde g} + 2 \tilde D^i Y_i.\] 
Together with Lemma \ref{lemma_Brown_York}, we obtain
\begin{theorem}\label{mass by Jang's equation}
\begin{align*}
E(\Sigma, \mathcal \mathcal X, T_0) &= \frac{1}{8\pi} \int \lt( \mu - J(\tilde e_4) \rt)d\tilde\Omega\\
&\quad + \frac{1}{16\pi} \int \lt[ \lt| k_{ij} - \frac{D_iD_ju}{\sqrt{1+|Du|^2}} \rt|^2_{\tilde g} + 2 |Y|^2_{\tilde g} \rt] d\tilde\Omega \\
&\quad + \frac{1}{16\pi} \int \lt[ |h_0(s) - h(s)|^2 - (H_0(s) - H(s))^2 \rt] d\tilde\Omega \\
&\quad + \frac{1}{8\pi} \int \lt[ |H|\sqrt{1+|\na\tau|^2} (\cosh\theta' - \cosh\theta) + \Delta\tau (\theta' - \theta)\rt] d\Sigma
\end{align*}
\end{theorem}

\section{Limit of the Wang-Yau mass}
In this section, we study the optimal embedding equation and the Jang equation with respect to the observer 
\[
T_0=(a_0,-a_1,-a_2,-a_3).
\]
Before restating our main result, Theorem \ref{main_theorem}, recall that in the vacuum case we view the initial data set $(M,g,k)$ as a spacelike hypersurface in its future development $N$ and denote the Weyl curvature of $N$ by $\bar W$. Let $e_0, e_1, e_2, e_3$ be an orthonormal basis at $p$ with $e_0$ the unit timelike normal of $M$ in $N$ and identify $T_0$ with the timelike vector $a_0 e_0 + \sum_{i=1}^3 a_i e_i$ at $p$.  
\begin{theorem}\label{main}
The quasi-local mass $E(\Sigma, \mathcal X, T_0)$ for  $T_0 = (a_0, -a_1, -a_2, -a_3)$ has the following asymptotic behavior as $d\rightarrow \infty $ for each of the following isometric embeddings $\mathcal X$ of $\Sigma$.
\begin{enumerate}
\item For the isometric embedding $\mathcal X:\Sigma \rightarrow \R^3$, we have
\begin{align*}
E(\Sigma, \mathcal X, T_0) = \frac{1}{6} \lt( a_0 \mu(p) - a_i J^i(p) \rt) + O(d^{-3-2\alpha}),
\end{align*} where $\mu$ and $J^i$ are defined in \eqref{Hamiltonian constraint} and \eqref{momentum constraint}. 
\item Suppose the initial data set $(M,g,k)$ satisfies the vacuum constraint equation \eqref{vc}.  For $(\mathcal X,T_0)$ solving the leading order of the optimal embedding equation \eqref{oee}, we have
\begin{align*}
E(\Sigma, \mathcal X, T_0) &= \frac{1}{90} \lt( Q(e_0,e_0,e_0,T_0) + \frac{1}{2a_0} \bar W_{0i0j} \bar W\indices{_0^i_0^j} \rt) + O(d^{-4-3\alpha}). 
\end{align*}
\end{enumerate}
Here $\bar W_{0i0j} = \bar W(e_0, e_i, e_0 ,e_j)(p)$ and $Q$ is the Bel-Robinson tensor of $N$ at $p$.
\end{theorem}

\begin{remark}
We mostly work at the initial data level and the error term has order $O(d^{-5-2\alpha})$. Only when the result  is expressed in terms of the spacetime curvature using the Gauss equation of $N$, $
\bar W_{0i0j} = R_{ij} + O(d^{-2-2\alpha}),$
does the error become $O(d^{-4-3\alpha})$.  
\end{remark}

The outline of this section is as follows. We solve the optimal embedding equation and the Dirichlet problem of Jang's equation in the first two subsections and then evaluate each integral in Theorem \ref{mass by Jang's equation} in the subsequent three subsections. Finally, we put everything together to prove Theorem \ref{main}.

\subsection{Optimal embedding}
Let us begin with the optimal embedding equation.
\begin{lemma}\label{optimal, general observer}
The following pair $T_0=(a_0,-a_1,-a_2,-a_3)$ and 
\[
\begin{split}
\mathcal X^0 = & \frac{1}{2}k_{ij}(p) \tilde x^i \tilde x^j + \frac{1}{6}\partial_i k_{jm}(p) \tilde x^i \tilde x^j \tilde x^m + \frac{a_i R_{mn}(p)\tilde x^m \tilde x^n \tilde x^i}{6a_0} \\
\mathcal X^i =  &  \tilde x^i - \frac{1}{6} R_{in}(p)\tilde x^n - \frac{1}{6} R_{mn}(p) \tilde x^m \tilde x^n \tilde x^i 
\end{split}
\]
solves the first two order of the optimal embedding equation. In particular, the above solution gives a time function $\tau = - \langle \mathcal{X},T_0 \rangle$ with
\begin{align}\label{time function}
\tau = a_i \tilde x^i + a_0 \left [  \frac{1}{2}k_{ij}(p) \tilde x^i \tilde x^j + \frac{1}{6}\partial_i k_{jm}(p) \tilde x^i \tilde x^j \tilde x^m \right ] - \frac{1}{6} a_i R^i_n(p) \tilde x^n  
\end{align}
\end{lemma}
\begin{proof}
With $T_0=(a_0,-a_1,-a_2,-a_3)$, the optimal embedding equation reads
\[
\frac{1}{2} \Delta (\Delta + 2 ) \mathcal{X}^0 =  div \alpha_H + \frac{a_i}{a_0} \left[  div ((H_0 - |H|) \nabla \tilde x^i)  + \frac{1}{2} \Delta ((H_0 - |H|) \tilde x^i)  \right ]+ O(d^{-3-\alpha}).
\]
See \cite[Section 7]{Chen-Wang-Yau3}. We first compute
\begin{align*}
(\alpha_H)_a &= -k(\pl_a, \nu) + \pl_a \lt( \frac{\mbox{tr}_\Sigma k}{H} \rt) + O(d^{-3-3\alpha})\\
&= -2 k_{ij} \tilde x^i_a \tilde x^j - 2 \pl_m k_{ij}  \tilde x^i_a \tilde x^j \tilde x^m + \frac{1}{2} \pl_m k_{ii} \tilde x^m_a -\frac{1}{2} \pl_m k_{ij} \tilde x^i \tilde x^j \tilde x^m_a +O(d^{-3-\alpha}).
\end{align*}
Using the Codazzi equation, $\pl_i k_{im} = \pl_m k_{ii} + O(d^{-3-2\alpha})$, we obtain
\begin{align*}
div \alpha_H &= -2 k_{ij} \lt( \delta^{ij} - 3 \tilde x^i \tilde x^j \rt) + 10 \pl_m k_{ij} \lt( \tilde x^i \tilde x^j \tilde x^m - \frac{1}{5} \delta^{ij} \tilde x^m - \frac{1}{5} \delta^{im} \tilde x^j - \frac{1}{5} \delta^{jm} \tilde x^i \rt) + O(d^{-3-\alpha})
\end{align*}
Note that $\delta^{ij} - 3 \tilde x^i \tilde x^j$ and $\tilde x^i \tilde x^j \tilde x^m - \frac{1}{5} \delta^{ij} \tilde x^m - \frac{1}{5} \delta^{im} \tilde x^j - \frac{1}{5} \delta^{jm} \tilde x^i$ are $-6$ and $-12$ eigenfunctions respectively.

On the other hand, by Lemma \ref{lemma_mean_curvature_diff_1} and $|H| = \sqrt{H^2 - (\mbox{tr}_\Sigma k)^2}$, we have 
\[ H_0 - |H| =  - R_{ij} \tilde x^i \tilde x^j + O(d^{-2-2\alpha})
\]
and 
\[
  div ((H_0 - |H|) \nabla \tilde x^i)  + \frac{1}{2} \Delta ((H_0 - |H|) \tilde x^i)  = 10 R_{mn} \tilde x^m \tilde x^n \tilde x^i - 4 R_{in} \tilde x^n + O(d^{-2-2\alpha})
\]
It follows that the given $\mathcal X^0$ satisfies the equation up to error of the order $ O(d^{-2-2\alpha})$.

For the $\mathcal X^i$, we use the well-known formula of Riemann curvature tensor in 3-dimension 
\begin{align}\label{riemann3d}
R_{ikjl} = g_{ij} R_{kl} - g_{il}R_{kj}+g_{kl}R_{ij}-g_{kj}R_{il} - \frac{R}{2} (g_{ij}g_{kl}-g_{il}g_{kj})
\end{align} and \eqref{Hamiltonian constraint} to show that the induced metric is
\[
\begin{split}
\sigma_{ab} =  \tilde \sigma_{ab} - \frac{1}{3} R_{kl}\tilde x^k_a \tilde x^l_b  -\frac{1}{3} R_{ij} \tilde x^i \tilde x^j \tilde\sigma_{ab} + O(d^{-2-2\alpha}).
\end{split}
\]
The lemma follows from the linearized isometric embedding equation into $\R^3$.
\end{proof}

\subsection{Jang's equation}
We work in local coordinates and $\Omega$ is identified with $B_1 \subset \R^3$. We discuss the solution of Dirichlet problem of Jang's equation
\begin{align}
\lt\{ \begin{array}{cl}
\lt( g^{ij} - \frac{D^iu D^ju}{1 + |D u|^2} \rt) \lt( k_{ij}-\frac{D_i D_j u}{\sqrt{1 + |Du|^2}} \rt)=0 & \mbox{ in } B_1\\
u = \tau & \mbox{ on } \pl B_1
\end{array} \rt.
\end{align}  
Let $b_{ij} = g_{ij} - \delta_{ij} + \frac{1}{3} R_{ikjl}(p) X^k X^l$ and $c_{ij} = k_{ij} - k_{ij}(p) - \pl_m k_{ij}(p) X^m$. By Definition \ref{initial_data} and \eqref{expansion_geodesic}, we have \begin{align}\label{weak_decay_assumption}
\| b_{ij}\|_{C^1 (B_1)}, \| c_{ij}\|_{C^1(B_1)} \le C d^{-3-\alpha}.
\end{align}

\begin{lemma}\label{solve_Jang}
Let $u$ be the solution of the Dirichlet problem of Jang's equation on $\Omega$ with boundary value $\tau$ given in \eqref{time function}. Then $u = a_i X^i + \frac{1}{6} T_i X^i + \frac{1}{2} a_0 k_{ij} X^i X^j +  \frac{1}{6} B_{ijm} X^iX^jX^m - \frac{1}{6} a_i R_{in} X^n + v$, where
\begin{align*}
T_m &= -\frac{1}{4 + \frac{1}{a_0^2}} \cdot \frac{2}{a_0}  a_i a_j \bar W\indices{_0^i_m^j}(p) + \frac{1}{4 + \frac{1}{a_0^2}} \cdot 2  a_l R^l_m(p) + \frac{1}{(4+\frac{1}{a_0^2})(2+ \frac{3}{a_0^2})}\cdot \frac{4a_m}{a_0}  a_i a_j R^{ij}(p)\\
B_{ijm} &= \frac{a_0}{3}(\pl_i k_{jm}(p) + \pl_j k_{im}(p) + \pl_m k_{ij}(p) ) - \frac{1}{3}(\delta_{ij} T_m + \delta_{jm} T_i + \delta_{mi} T_j),
\end{align*}
and $\| v \|_{C^{2,\beta}} \le C' d^{-3-\alpha}$ for some constant $C'$ and $0 < \beta < 1$ depending only on $C$ in \eqref{weak_decay_assumption}.
\end{lemma}
\begin{proof}
We write $u = a_i X^i + b_i X^i + \frac{1}{2} a_0 k_{ij} X^i X^j + \frac{1}{6} B_{ijm}X^iX^jX^m - \frac{1}{6} a_i R_{in} X^n + v$ for constants $b_i, B_{ijm}$ to be determined from the leading order of Jang's equation.

Define a symmetric 3-tensor
\[ T_{ijm} = \frac{a_0}{3}(\pl_i k_{jm} + \pl_j k_{mi} + \pl_m k_{ij}) - B_{ijm}.\]

We need to show that $\tau-u$ is perpendicular to all $-2$ and $-12$ eigenfunctions on $\pl B_1=S^2$ which,  by the expression \eqref{time function} of $\tau$, is equivalent to  
\begin{align*}
&\int_{S^2} \lt( \frac{1}{6}T_{ijm} \tilde x^i \tilde x^j \tilde x^m - b_i \tilde x^i \rt) \tilde x^l dS^2=0, \\
&\int_{S^2} T_{ijm} \tilde x^i \tilde x^j \tilde x^m (\tilde x^l \tilde x^p \tilde x^q - \frac{1}{5} \tilde x^l \delta^{pq} - \frac{1}{5} \tilde x^p \delta^{ql} - \frac{1}{5} \tilde x^q \delta^{lp}) dS^2 =0.
\end{align*}
The second equation contains 7 linear equations with 10 variables. Using  \cite[Lemma 5.3]{Chen-Wang-Yau4}, we solve $T_{ijm}$ by free variables $T_{111}, T_{222}, T_{333}$:
\[ T_{ijm} = \frac{1}{3} (\delta_{ij} T_m + \delta_{jm} T_i + \delta_{mi} T_j), \quad T_m := T_{mmm}. \]
and then solve $b_i = \frac{1}{6} T_i$. 

It remains to solve $T_m$ from Jang's equation
\[ (g^{ij} - \frac{D^iu D^ju}{1 + |Du|^2}) (k_{ij} - \frac{D_iD_j u}{\sqrt{1 + |Du|^2}} ) =0. \]
By the Codazzi equation, $\pl_m k_{ij} - \pl_i k_{mj} = - \bar W_{0jmi} + O(d^{-3-2\alpha})$, it follows that 
\[ k_{ij} - \frac{u_{ij}}{\sqrt{1+|Du|^2}} = S_{ijm}X^m + O(d^{-3-\alpha}) \]
where
\begin{align}\label{S}
S_{ijm}=\frac{1}{3} \lt( \frac{1}{a_0}( \delta_{ij} T_m + \delta_{jm}T_i - \delta_{mi} T_j) - \bar W_{0imj} - \bar W_{0jmi}+ \frac{a_l}{a_0}(R_{iljm} + R_{imjl})\rt).
\end{align}
The leading order of Jang's equation thus reduces to 3 linear equations \[ \sum_{i,j} (\delta_{ij} - \frac{a_ia_j}{a_0^2}) S_{ijm} =0.\]
We have 
\begin{align*}
0 &= \sum_{i,j} (\delta_{ij} - \frac{a_i a_j}{a_0^2})( \delta_{ij} T_m + \delta_{jm}T_i + \delta_{mi} T_j ) + \frac{a_ia_j}{a_0} (\bar W\indices{_0^i_m^j} + \bar W\indices{_0^j_m^i}) - 2 a_l R^l_m \\
&= \sum_i G_{im} T_i + \frac{a_ia_j}{a_0} (\bar W\indices{_0^i_m^j} + \bar W\indices{_0^j_m^i}) - 2 a_l R^l_m, 
\end{align*}
where $G_{im} = (4 + \frac{1}{a_0^2} )\delta_{im} -2 \frac{a_ia_m}{a_0^2}$. We solve for the inverse matrix of $G_{im}$ \[ ( G^{-1})_{ml} = \frac{1}{(4+ \frac{1}{a_0^2}) (2 + \frac{3}{a_0^2})} \lt[ (2 + \frac{3}{a_0^2}) \delta_{ml} + 2 \frac{a_ma_l}{a_0^2}\rt] \]
to get
\[ T_m = -\frac{1}{4 + \frac{1}{a_0^2}} \cdot \frac{2a_ia_j}{a_0} \bar W\indices{_0^i_m^j} + \frac{1}{4 + \frac{1}{a_0^2}} \cdot 2 a_l R^l_m + \frac{1}{(4+\frac{1}{a_0^2})(2+ \frac{3}{a_0^2})}\cdot \frac{4a_m}{a_0^2} R^{ij}a_i a_j. \]

After obtaining the leading order of $u$, we treat Jang's equation as a quasilinear partial differential equation $Q(x, Dv)v=0$ in $B_1$ and $v=0$ on $\pl B_1$. By \eqref{weak_decay_assumption}, we can choose a constant $C'$ that depends only on $C$ such that $ \pm C' d^{-3-\alpha} \lt( |x|^2 - 1 \rt)$ is a sub/super solution to this equation. This provides the $C^0$-estimate and the boundary gradient estimate. By \cite[Theorem 15.1]{GT} and \cite[Theorem 13.7]{GT}, we get the gradient estimate and H\"{o}lder estimate for the gradient. The $C^{2,\beta}$ a priori estimate for $v$ and the solvability of $u$ then follows from the Schauder estimate and \cite[Theorem 11.4]{GT}.
\end{proof}

\begin{proof}[Proof of Theorem \ref{main}, (1)]
We first examine the limit for an initial data set with matter fields. We will show below that all terms except $\mu - J(\tilde e_4)$ are of the order $O(d^{-4-2\alpha})$. As a result,
\[
E(\Sigma) = \frac{1}{8\pi} \int (\mu - J(\tilde e_4)) d\tilde\Omega + O(d^{-4-2\alpha}). 
\]
The assertion follows from $d\tilde\Omega = a_0 dx + O(d^{-2-\alpha})$ and $\tilde e_4 = \frac{1}{a_0}(-1,a_1,a_2,a_3) + O(d^{-1-\alpha})$.
\end{proof}

For a vacuum initial data set, Theorem \ref{mass by Jang's equation} becomes
\begin{align}
E(\Sigma, X , T_0 ) &= \frac{1}{16\pi} \int \lt[ \lt| k_{ij} - \frac{D_iD_ju}{\sqrt{1+|Du|^2}} \rt|^2_{\tilde g} + 2 |Y|^2_{\tilde g} \rt] d\tilde\Omega \label{Schoen-Yau}\\
&\quad + \frac{1}{16\pi} \int \lt[ |h_0(s) - h(s)|^2 - (H_0(s) - H(s))^2 \rt] d\tilde\Omega \label{difference_2nd_ff}\\
&\quad + \frac{1}{8\pi} \int \lt[ |H|\sqrt{1+|\na\tau|^2} (\cosh\theta' - \cosh\theta) + \Delta\tau (\theta' - \theta)\rt] d\Sigma \label{gauge}
\end{align}

\subsection{Evaluation of \eqref{Schoen-Yau}}
\begin{lemma}\label{lemma_Schoen-Yau}
\begin{align*}
   &\frac{1}{16\pi} \int_{B_1} \lt( \lt| k_{ij} - \frac{D_iD_ju}{\sqrt{1 + |Du|^2}} \rt|^2_{\bar g} + 2 |Y|^2_{\bar g} \rt) d\tilde\Omega  \\
=& -\frac{1}{54a_0} \sum_{m} T_m^2 - \frac{1}{27a_0}  T_m a_l R^{lm}(p) + \frac{1}{90} \lt( a_0 - \frac{1}{a_0} \rt) R_{ij}(p) R^{ij}(p) - \frac{1}{54a_0} \sum_{i,j,l} a_i a_j R^{il}(p) R^j_l(p)\\
& + \frac{a_0}{180} \bar W_{0imj}(p) \bar W_0^{\;\;imj}(p) - \frac{1}{45} a_m \bar W\indices{_0^{imj}}(p) \bar W_{0i0j}(p) + O(d^{-5-2\alpha}).
\end{align*}
\end{lemma}
\begin{proof}
Recall that we write $k_{ij} - \frac{D_iD_j u}{\sqrt{1+|Du|^2}} = S_{ijm} X^m + O(d^{-3-\alpha})$ in the proof of Lemma \ref{solve_Jang}. Since $\tilde e_4 = \frac{(-1, a_1, a_2, a_3)}{a_0}+ O(d^{-1-\alpha})$, we have $Y = \sum_{i,m,p} \frac{a_p}{a_0} S_{ipm} X^m \frac{\pl}{\pl X^i} + O(d^{-3-\alpha})$ and   
\begin{align*}
&\frac{1}{16\pi} \int_{B_1} |k_{ij} - \frac{D_iD_ju}{\sqrt{1 + |Du|^2}}|^2_{\bar g} + 2 |Y|^2_{\bar g} dV_{\bar g} \\
 =& \frac{a_0}{60} \sum_{i,j,m,p,q} \lt[ (\delta_{ij} - \frac{a_i a_j}{a_0^2}) S_{ipm} S_{jqm} (\delta_{pq} - \frac{a_p a_q}{a_0^2}) + 2 (\delta_{ij} - \frac{a_i a_j}{a_0^2}) \frac{a_pa_q}{a_0^2} S_{ipm} S_{jqm} \rt]+ O(d^{-5-2\alpha})\\
=& \frac{a_0}{60} \lt[ \sum_{i,j,m} S_{ijm}^2 - \sum_{i,j,m,p,q} \frac{a_i a_j a_p a_q}{a_0^4} S_{ipm} S_{jqm} \rt] + O(d^{-5-2\alpha})\\
= &\frac{a_0}{60} \lt[ \sum_{i,j,m} S_{ijm}^2  - \sum_m \lt( \sum_i S_{iim} \rt)^2 \rt]+ O(d^{-5-2\alpha}),
\end{align*}
where Jang's equation is used in the last equality.
We compute, by (\ref{S}), 
\begin{align*}
\sum_m \lt( \sum_i S_{iim} \rt)^2= \frac{1}{9a_0^2} \sum_m (5 T_m + 2a_l R^l_m)^2 = \frac{1}{9a_0^2} \lt( 25 \sum_m T_m^2 + 10 a_l T_m R^{lm} + 4 a_l a_n R^{lm}R^n_m \rt)
\end{align*}
and
\begin{align*}
\sum_{i,j,m} S_{ijm}^2 &= \frac{1}{9} \Big[ \frac{15}{a_0^2} \sum_m T_m^2 + \sum_{i,j,m}(\bar W_{0imj} + \bar W_{0jmi})^2 \\
&+ \frac{1}{a_0^2} \sum_{i,j,m}\lt( 2 a_m R_{ij} + 2 a_l R^l_m\delta_{ij} - a_j R_{im} - a_i R_{jm} -a_l R^l_j \delta_{im} - a_l R^l_i R_{jm}\rt)^2 \\
&- 2 \sum_{i,j,l,m}\frac{a_l}{a_0}(\bar W_{0imj} + \bar W_{0jmi})(R_{iljm} + R_{imjl}) \Big]
\end{align*}
By the first Bianchi identity, $2 \sum_{i,j,m}\bar W_{0imj} \bar W_{0jmi} = \sum_{i,j,m} \bar W_{0imj} \bar W_{0imj}$ and hence $\sum_{i,j,m} (\bar W_{0imj} + \bar W_{0jmi})^2 = 3 \sum_{i,j,m} \bar W_{0imj}^2$. Direct computation shows that the third term in the bracket is equal to $\frac{1}{a_0^2} (6 \sum_m a_m^2 R_{ij}^2 - 6 a_i a_j R^{il} R^j_l)$. Finally, by (\ref{riemann3d}), the last term in the bracket is equal to $-12 \sum_{i,j,m} \frac{a_m}{a_0} \bar W_{0imj} \bar W_{0i0j}$.
\end{proof}

\subsection{Evaluation of \eqref{difference_2nd_ff}}
\begin{lemma}\label{lemma_difference_2nd_ff}
\begin{align*}
\frac{1}{16\pi} \int \lt[ |h_0(s) - h(s)|^2 - (H_0(s) - H(s))^2 \rt] d\tilde\Omega = \frac{1}{60a_0} R_{ij}(p)R^{ij}(p) + O(d^{-5-2\alpha}).
\end{align*}
\end{lemma}
\begin{proof}
By Lemma \ref{difference of second fundamental forms} below, we have
\[ h_0(s) - h(s) = s^3 \frac{\sqrt{1 + |a_i \tilde\na \tilde x^i|^2}}{a_0} R_{ij} \tilde x^i_a \tilde x^j_b + O(d^{-3-\alpha}). \]
We compute
\begin{align*}
|h_0(s) - h(s)|^2_{\sigma} - \lt( H_0(s) - H(s)\rt)^2 &= (\mbox{det}_{\tilde\sigma}(\sigma))^{-1} \lt( |h_0(s) - h(s)|^2_{\tilde\sigma} - (\mbox{tr}_{\tilde\sigma} h_0(s) - \mbox{tr}_{\tilde\sigma} h(s))^2 \rt) \\ &= \frac{s^2}{a_0^2} \lt( R_{ij}R^{ij} - 2 R_{ij} R^i_l \tilde x^j \tilde x^l \rt) + O(d^{-5-2\alpha}).
\end{align*}
Finally, we note that the volume form $d\tilde\Omega = a_0 s^2 dS^2 ds + O(d^{-2-\alpha})$ and hence
\begin{align*}
  &\int \lt[ |h_0(s) - h(s)|^2 - (H_0(s) - H(s))^2 \rt]d\tilde\Omega  \\
=& \frac{1}{a_0} \int_0^1 \int_{S^2} \lt( R_{ij}R^{ij} - 2 R_{ij} R^i_l \tilde x^i \tilde x^j \rt) dS^2 s^4 ds+ O(d^{-5-2\alpha})\\
=& \frac{4}{15a_0} R_{ij} R^{ij}+ O(d^{-5-2\alpha}).
\end{align*}
\end{proof}
The rest of this subsection is devoted to computing the difference of second fundamental forms of $\Sigma_s$ in $\tilde\Omega$ and in $\R^3$, Lemma \ref{difference of second fundamental forms}. We first solve the isometric embedding of $\Sigma_s$ into $\R^3$ and then compute the second fundamental form. Recall the solution of Jang's equation is $u = a_i X^i + \frac{1}{2} a_0 k_{ij} X^i X^j + \frac{1}{6} B_{ijm} X^i X^j X^m - \frac{a_i}{6} R_{in} X^n + O(d^{-4})$. The induced metric on the graph of Jang's equation is then given by
\begin{align*}
\bar g_{ij} &= \delta_{ij} + a_i a_j -\frac{1}{3} R_{ikjl} X^k X^l + a_0 (a_i k_{jm} + a_j k_{im} ) X^m\\
&\quad  + \frac{1}{2} (a_i B_{jlm} + a_j B_{ilm}) X^l X^m - \frac{1}{6} a_i a_l R_{lj} - \frac{1}{6} a_j a_l R_{li} +  O(d^{-3 - \alpha}).
\end{align*}

In polar coordinates $(s,u^a)$, we have $\bar g = \bar g_{ss} ds^2 + 2 \bar g_{as} ds du^a + \bar g_{ab} du^a du^b$ where
\begin{align*}
\bar g_{ss} &= 1 + (a_i \tilde x^i)^2\\
&\quad + a_i \tilde x^i \lt( 2s a_0 k_{jm} \tilde x^j \tilde x^m + s^2  B_{jlm} \tilde x^j \tilde x^l \tilde x^m - \frac{1}{3} a_l R_{lj} \tilde x^j \rt)+ O(d^{-3 - \alpha}).\\
\bar g_{as} &= s a_i \tilde x^i_a a_j \tilde x^j + s a_i \tilde x^i_a \lt( s a_0 k_{jm} \tilde x^j \tilde x^m + \frac{1}{2}s^2  B_{jlm} \tilde x^j \tilde x^l \tilde x^m - \frac{1}{6} a_l R_{lj} \tilde x^j \rt) \\
&\quad + s a_j \tilde x^j \lt( sa_0 k_{im} \tilde x^i_a \tilde x^m + \frac{1}{2} s^2 B_{ilm} \tilde x^i_a \tilde x^l \tilde x^m - \frac{1}{6} a_l R_{li} \tilde x^i_a \rt)+ O(d^{-3 - \alpha}).\\
\bar g_{ab} &= s^2 \tilde\sigma_{ab} + s^2 \lt( a_i a_j - \frac{s^2}{3} R_{ikjl} \tilde x^k \tilde x^l \rt) \tilde x^i_a \tilde x^j_b +s^3 a_0(a_i k_{jm} + a_j k_{im}) \tilde x^i_a \tilde x^j_b \tilde x^m  \\
&\quad + \frac{1}{2} s^4 (a_iB_{jlm} + a_j B_{ilm}) \tilde x^i_a \tilde x^j_b \tilde x^l \tilde x^m - \frac{1}{6}s^2(a_i a_l R^l_j + a_j a_l R^l_i) \tilde x^i_a \tilde x^j_b  + O(d^{-3 - \alpha}).
\end{align*}
Let $\sigma(s)$ be the induced metric on $\Sigma_s$. We consider the isometric embedding of $(\Sigma_s, \sigma(s)_{ab})$ into the hyperplane $X^0 = a_i X^i$ in $\R^4$ with the form $X^0 = a_i s( \tilde x^i + \mathfrak{y}^i)$ and $X^i = s (\tilde x^i + \mathfrak{y}^i)$ where $\mathfrak{y}^i = O(d^{-1-\alpha})$ and satisfies the linearized isometric embedding equations
\begin{align*}
(\delta_{ij} + a_i a_j)(\tilde x^i_a \mathfrak{y}^j_b + \tilde x^j_b \mathfrak{y}^i_a) &= -\frac{s^2}{3} R_{ikjl} \tilde x^k \tilde x^l \tilde x^i_a \tilde x^j_b \\
&\quad + s a_0 (a_i k_{jm} + a_j k_{im}) \tilde x^m \tilde x^i_a \tilde x^j_b + \frac{s^2}{2} (a_i B_{jlm} + a_j B_{ilm}) \tilde x^l \tilde x^m \tilde x^i_a \tilde x^j_b\\
&\quad -\frac{1}{6}(a_i R^l_j + a_j R^l_i) a_l \tilde x^i_a \tilde x^j_b + O(d^{-3-\alpha}).
\end{align*}
It's not hard to see that $\mathfrak{y}^i = y^{(0)i} + y^{i}$ with
\begin{align*}
y^{(0)i} = a_i \lt( \frac{s}{2a_0}  k_{lm} \tilde x^l \tilde x^m + \frac{s^2}{6a_0^2} B_{jlm} \tilde x^j \tilde x^l \tilde x^m - \frac{1}{6a_0^2} a_m R^m_n \tilde x^n  \rt)
\end{align*}
and 
\begin{align}\label{iso_ellipsoid} (\delta_{ij} + a_i a_j)(\tilde x^i_a y^j_b + \tilde x^j_b y^i_a) = -\frac{s^2}{3} R_{ikjl} \tilde x^k \tilde x^l \tilde x^i_a \tilde x^j_b. \end{align}
Equation (\ref{iso_ellipsoid}) is a linearized isometric embedding equation on an ellipsoid. Let $\hat y^i = (\delta_{ij} + a_i a_j) y^j$. One readily verifies that
\[ \hat y^i = -\frac{s^2}{6} \lt( R^i_j \tilde x^j + R_{jk} \tilde x^j \tilde x^k \tilde x^i \rt) \]
solves \eqref{iso_ellipsoid}. 

The family of isometric embedding of $\Sigma_s$ forms a foliation $F: (0,1] \times S^2 \rightarrow \{ X^0 = a_i X^i\} \subset \R^4$. From $F_*(\pl_s) = (\tilde x^i + \mathfrak{y}^i + s \frac{\pl\mathfrak{y}^i}{\pl s}) (a_i \frac{\pl}{\pl X^0} + \frac{\pl}{\pl X^i})$ and $F_*(\pl_a) = s(\tilde x^i_a + \mathfrak{y}^i_a)(a_i \frac{\pl}{\pl X^0} + \frac{\pl}{\pl X^i})$, we could write the flat metric, denoted by $\breve g$, in $(s,u^a)$ coordinates. Straightforward computation shows that 
\begin{align*}
\breve g_{ss} &= \bar g_{ss} + -2s^2 R_{ij} \tilde x^i \tilde x^j + O(d^{-3-\alpha})\\
\breve g_{bs} &= \bar g_{bs} -s^3 R_{ij} \tilde x^i_b \tilde x^j + O(d^{-3-\alpha}).
\end{align*}
\begin{remark}
It shouldn't be surprising that $\breve g - \bar g$ does not depend on $y^{(0)}$; namely, $\bar g$ remains flat after the graphical perturbation of $u$.  Indeed, if $\bar g_{ij} = g_{ij} + u_i u_j$, then the curvature tensors are related by
\[ \bar R_{ij\;\;l}^{\;\;\;k} = R_{ij\;\;l}^{\;\;\;k} - \frac{1}{1 + |\na u|^2} R_{ij\;\;l}^{\;\;\;p} \na_p u \na^k u + X_{ij\;\;l}^{\;\;\;k}, \]
where $X$ is quadratic in the Hessian of $u$.
\end{remark} 

We are ready to compute the second fundamental forms.
\begin{lemma}\label{difference of second fundamental forms}
\[ \bar h_{ab} - \breve h_{ab} = -s^3 \frac{\sqrt{1 + |a_i \tilde\na \tilde x^i|^2}}{a_0} R_{ij}(p) \tilde x^i_a \tilde x^j_b + O(d^{ -3-\alpha}). \]
\end{lemma}
\begin{proof}
We denote the leading order of $\tau$ by $\bar \tau = a_i \tilde x^i$. The second fundamental form of $\Sigma_s$ can be computed from the formula
\[ \bar h = \frac{\frac{1}{2} \pl_s \bar g_{ab} - \na_a \bar g_{bs}}{\sqrt{\bar g_{ss} - \sigma^{ab}\bar g_{as} \bar g_{bs}}}. \]
Here $\sigma$ is the induced metric of $\Sigma_s$. We note that $\sigma_{ab} = s^2(\tilde\sigma_{ab} + \bar \tau_a \bar\tau_b) + O(d^{-1-\alpha})$.

We compute
\begin{align*}
\bar g_{ss} - \sigma^{ab} \bar g_{as} \bar g_{bs} &= \frac{a_0^2}{ 1 + |\tilde\na\bar \tau|^2} + O(d^{-2-\alpha}) \\
\frac{1}{2} \pl_s \bar g_{ab} - \na_a \bar g_{bs} &= s \frac{a_0^2}{1 + |\tilde\na\bar\tau|^2} \tilde\sigma_{ab} + O(d^{-2-\alpha}).
\end{align*}
Hence, the difference of second fundamental forms is given by
\begin{align*}
\bar h - \breve h &= \lt( \frac{1}{2} \pl_s \bar g_{ab} - \na_a \bar g_{bs} \rt) \cdot \frac{\lt( \breve g_{ss} - \sigma^{cd} \breve g_{cs} \breve g_{ds} \rt) - \lt( \bar g_{ss} - \sigma^{cd} \bar g_{cs} \bar g_{ds} \rt)}{2 a_0^3} \cdot ( 1+|\tilde\na\bar \tau|^2)^\frac{3}{2} \\
&\quad + \frac{\na_a (\breve g_{bs} - \bar g_{bs})}{a_0} \sqrt{1+|\tilde\na\bar \tau|^2} + O(d^{-4-2\alpha})\\
&= \frac{\sqrt{1+|\tilde\na\bar\tau|^2}}{a_0} \lt( \frac{1}{2}( \breve g_{ss} - \bar g_{ss} + \sigma^{cd} \bar g_{cs} \bar g_{ds} - \sigma^{cd}\breve g_{cs} \breve g_{ds}) \cdot s \tilde \sigma_{ab} + \na_a(\breve g_{bs} - \bar g_{as})\rt) + O(d^{-3-\alpha}).
\end{align*}

We compute
\begin{align*}
\na_a (\breve g_{bs} - \bar g_{bs}) &= \tilde\na_a (\breve g_{bs} - \bar g_{bs}) - \frac{\tilde\na^c \bar \tau}{1 + |\tilde\na \bar \tau|^2} \tilde\na_a \tilde \na_b \bar \tau (\breve g_{cs} - \bar g_{cs}) + O(d^{-3-2\alpha})\\
&= s^3 R_{ij} \tilde x^i \tilde x^j \tilde\sigma_{ab} - s^3 R_{ij} \tilde x^i_a \tilde x^j_b  - \frac{\tilde\na^c \bar \tau}{1 + |\tilde\na \bar \tau|^2} \tilde\na_a \tilde \na_b \bar \tau (\breve g_{cs} - \bar g_{cs}) + O(d^{ -3-\alpha}).
\end{align*}
On the other hand,
\begin{align*}
\frac{1}{2} \lt( \breve g_{ss} - \bar g_{ss} + \sigma^{ab}(\bar g_{as}\bar g_{bs} - \breve g_{as} \breve g_{bs}) \rt) &= -s^2 R_{ij} \tilde x^i \tilde x^j + \frac{1}{2} \sigma^{cd} \lt( \bar g_{cs} \lt( \bar g_{ds} - \breve g_{ds} \rt) + \breve g_{ds} \lt( \bar g_{cs} - \breve g_{cs} \rt) \rt) +  O(d^{ -3-\alpha})\\
&= -s^2 R_{ij} \tilde x^i \tilde x^j + \frac{\tilde\na\bar\tau}{1 + |\tilde\na\bar\tau|^2} \cdot s^{-1} \bar\tau (\bar g_{cs} - \breve g_{cs}) + O(d^{-3-\alpha}).
\end{align*}
Putting these together, the assertion follows.
\end{proof}
\subsection{Evaluation of \eqref{gauge}}
In this subsection, we evaluate the integral resulted from the difference between the gauge induced by Jang's equation and the canonical gauge. Recall the solution of optimal isometric embedding equation is $\tau = a_i \tilde x^i + \frac{a_0}{2} k_{ij} \tilde x^i \tilde x^j + \frac{a_0}{6} \pl_i k_{jm} \tilde x^i\tilde x^j \tilde x^m - \frac{a_i}{6} R_{in} \tilde x^n$ and the solution of Jang's equation is $u = a_i X^i +\frac{1}{6} T_i \tilde x^i+  \frac{1}{2} a_0 k_{ij} X^i X^j + \frac{1}{6} B_{ijm} X^i X^j X^m - \frac{a_i}{6} R_{in} X^n$. We will need the following lemma.
\begin{lemma}
For $\tau = a_i \tilde x^i + v$, $v =  O(d^{-1-\alpha})$, we have
\begin{align*}
\Delta\tau = \tilde\Delta \tau + \tilde\Delta v - \frac{2}{3} a_i R_{ij} \tilde x^j + \frac{1}{3} a_i R_{kl} \tilde x^i \tilde x^k \tilde x^l  + O(d^{-3-\alpha})  
\end{align*}
\end{lemma}
\begin{lemma}\label{lemma_gauge}
\begin{align*}
&\frac{1}{8\pi} \int \lt[ |H|\sqrt{1+|\na\tau|^2} (\cosh\theta' - \cosh\theta) + \Delta\tau (\theta' - \theta)\rt] d\Sigma \\
&= \frac{1}{54a_0} \lt( \sum_m T_m^2 + 2  T_m a_i R^{im}(p) + a_i a_j R^{il}(p) R^j_l(p) \rt) + O(d^{-5-2\alpha}).
\end{align*}
\end{lemma}
\begin{proof}
For the canonical gauge, we have
\begin{align*}
\langle H, \bar e_4 \rangle &= |H| \sinh\theta = \frac{-\Delta\tau}{\sqrt{1+|\na\tau|^2}} \\
&= \frac{1}{\sqrt{1 + |\na\tau|^2}} \Big[ 2 a_i \tilde x^i + 3 a_0 k_{ij}(\tilde x^i \tilde x^j - \frac{1}{3} \delta^{ij}) + a_0 \pl_i k_{jm} \lt( 2 \tilde x^i \tilde x^j \tilde x^m - \frac{1}{3} (\tilde x^i \delta^{jm} + \tilde x^j \delta^{mi} + \tilde x^m \delta^{ij}) \rt) \\
&\qquad + \frac{a_i}{3} R^i_n \tilde x^n - \frac{1}{3} a_m R_{ij} \tilde x^i \tilde x^j \tilde x^m \Big] + O(d^{-3-\alpha}).
\end{align*}

For the gauge $\{ e_3', e_4'\}$ induced by Jang's equation, we have \cite[Theorem 4.1]{Wang-Yau1} $e_4' = \sinh\phi e_3 + \cosh\phi e_4$ with 
\[ 
\begin{split}
\sinh\phi & = \frac{-u_3}{\sqrt{1+|\na\tau|^2}} \\
&= \frac{-1}{\sqrt{1+|\na\tau|^2}} (a_i \tilde x^i + \frac{1}{6} T_i \tilde x^i + a_0 k_{ij} \tilde x^i \tilde x^j + \frac{1}{2} B_{ijm} \tilde x^i \tilde x^j \tilde x^m - \frac{1}{6} a_i R^i_n \tilde x^n) + O(d^{-3-\alpha})
\end{split}
\]
and hence
\begin{align*}
   & \langle H, e_4' \rangle \\
=& \cosh\phi \langle H, e_4 \rangle + \sinh\phi \langle H,e_3 \rangle \\
=& \frac{a_0}{\sqrt{1+|\na\tau|^2}} \lt( -k_{ii} + k_{ij} \tilde x^i \tilde x^j - \pl_m k_{ii} \tilde x^m + \pl_m k_{ij} \tilde x^m \tilde x^i \tilde x^j\rt)\\
& +\frac{1}{\sqrt{1+|\na\tau|^2}} \lt( a_i \tilde x^i + \frac{1}{6} T_i \tilde x^i +a_0 k_{ij} \tilde x^i \tilde x^j + \frac{1}{2} B_{ijm} \tilde x^i \tilde x^j \tilde x^m - \frac{1}{6} a_i R^i_n \tilde x^n \rt) (2 - \frac{1}{3} R_{ij} \tilde x^i \tilde x^j) + O(d^{-3-\alpha}).
\end{align*}
By the constraint equations, we get
\[ |H| (\sinh \theta - \sinh \theta') = \frac{1}{\sqrt{1+|\na\tau|^2}} \frac{2}{3} \lt( T_m \tilde x^m + a_i R^i_j \tilde x^j \rt) + O(d^{-3-\alpha}). \]

Next, using two elementary computations 
\[ \cosh\theta' - \cosh\theta = \frac{\sinh\theta' + \sinh\theta}{\cosh\theta' + \cosh\theta} (\sinh\theta' - \sinh\theta)\]
and (up to error of order $O(d^{-6-3\alpha})$)
\begin{align*}
\theta' - \theta = \sinh(\theta' - \theta) = \sinh\theta' \cosh\theta - \sinh\theta \cosh\theta' = \cosh\theta (\sinh\theta' - \sinh\theta)-\sinh\theta (\cosh\theta' - \cosh\theta), 
\end{align*}
we get, up to a negligible error,
\begin{align*}
& |H|\sqrt{1+|\na\tau|^2} (\cosh\theta' - \cosh\theta) + \Delta\tau (\theta' - \theta)\\
=& \lt( \sqrt{1+|\na\tau|^2 }|H| - \sinh\theta \Delta\tau \rt) (\cosh\theta' - \cosh\theta) + \Delta\tau \cosh\theta (\sinh\theta' - \sinh\theta)\\
=& \lt( \lt( \sqrt{1 + |\na\tau|^2} |H| - \sinh\theta \Delta\tau \rt) \frac{\sinh\theta + \sinh\theta'}{\cosh\theta + \cosh\theta'} + \Delta\tau \cosh\theta \rt)(\sinh\theta' - \sinh\theta) \\
= &\sqrt{1+|\na\tau|^2}|H| \lt( \cosh^2\theta \frac{\sinh\theta + \sinh\theta'}{\cosh\theta + \cosh\theta'} - \sinh\theta \cosh\theta \rt)(\sinh\theta' - \sinh\theta)\\
= & \frac{1}{\cosh\theta + \cosh\theta'} \sqrt{1 + |\na\tau|^2} |H| (\sinh\theta' - \sinh\theta)^2  
\end{align*}

Consequently, 
\begin{align*}
&\frac{1}{8\pi} \int \lt[ |H|\sqrt{1+|\na\tau|^2} (\cosh\theta' - \cosh\theta) + \Delta\tau (\theta' - \theta)\rt] d\Sigma \\
=& \frac{1}{8\pi} \int \frac{1}{2\cosh\theta} \sqrt{1+|\na\tau|^2} |H| (\sinh\theta' - \sinh\theta)^2 dS^2\\
=& \frac{1}{8\pi} \frac{1}{4a_0} \frac{4}{9}\int \lt( T_m \tilde x^m +  a_i R^i_j \tilde x^j \rt)^2  dS^2 + O(d^{-5-2\alpha})\\
=& \frac{1}{54a_0} \lt( \sum_m T_m^2 + 2 T_m a_i R^{im} + a_i a_j R^{il} R^j_l \rt) + O(d^{-5-2\alpha}).
 \end{align*}
\end{proof}

\begin{proof}[Proof of Theorem \ref{main}, (2)]
Putting Lemma \ref{lemma_Schoen-Yau}, \ref{lemma_difference_2nd_ff}, \ref{lemma_gauge} together, we get
\begin{align*}
E(\Sigma,X,T_0) &=  \frac{a_0}{90} R_{ij}(p) R^{ij}(p) + \frac{1}{180a_0} R_{ij}(p) R^{ij}(p) \\
& + \frac{a_0}{180} \bar W_{0imj}(p) \bar W_0^{\;\;imj}(p) - \frac{1}{45} a_m \bar W\indices{_0^{imj}}(p) \bar W_{0i0j}(p).
\end{align*}
Recall \cite[page 4]{Chen-Wang-Yau4}
\begin{align*}
Q(e_0,e_0,e_0,T_0) = \frac{a_0}{2} \bar W_{0imj}(p) \bar W_0^{\;\;imj}(p) + a_0 \bar W_{0m0n}(p) \bar W\indices{_0^m_0^n}(p) - 2 a_m \bar W_0^{\;\;imj}(p)\bar W_{0i0j}(p).
\end{align*}
By the Gauss equation, $\bar W_{0i0j} = R_{ij} + O(d^{-2-2\alpha})$, we complete the proof.
\end{proof}
\section{Small sphere limit}
The careful readers would surely find the similarity between our main result and the small sphere limit \cite[Theorem 1.1, 1.2]{Chen-Wang-Yau4}. In this section, we adapt the previous computations to the small sphere setting. Although the family of small spheres is different from that considered in \cite{Chen-Wang-Yau4}, the limit of quasi-local mass turns out to be the same. 

Let $p$ be a point in the spacetime. We recall the setup in \cite{Chen-Wang-Yau4}. Let $e_0, e_1, e_2, e_3$ be an orthonormal basis at $p$, $\langle e_\alpha, e_\beta \rangle = \eta_{\alpha\beta}$. Using $e_0$, we normalize each null vector $L$ at $p$ by $\langle L, e_0 \rangle = -1$. We consider the null geodesics with initial velocity being the normalized $L$. $\Sigma_r$ is defined as the level sets of the affine parameter $r$. In short, Chen-Wang-Yau considered small spheres approaching $p$ along the light cone.
\begin{theorem}\label{small_sphere_limit}\cite[Theorem 1.1, 1.2]{Chen-Wang-Yau4}
\begin{enumerate}
\item For the isometric embeddings $\mathcal X_r$ of $\Sigma_r$ into $\R^3$, the quasi-local energy satisfies
\begin{align}\label{CWY_nonvacuum}
E(\Sigma_r, \mathcal X_r, T_0) = r^3 \cdot \frac{4\pi}{3} T(e_0,T_0) + O(r^4)
\end{align}
as $r$ goes to $0$.
\item Suppose the stress-energy tensor $T_{\alpha\beta}$ vanish in a neighborhood of $p$. Then, for the pair $(\mathcal X_r(T_0), T_0)$ solving the leading order term of the optimal embedding equation of $\Sigma_r$, we have
\begin{align}\label{CWY_vacuum}
E(\Sigma_r, \mathcal X_r(T_0), T_0) = r^5 \cdot \frac{1}{90} \lt[ Q(e_0,e_0,e_0, T_0) + \frac{\sum_{m,n} \bar W^2_{0m0n}(p)}{2a_0} \rt] + O(r^6)
\end{align}
as $r$ goes to $0$.
\end{enumerate}
In the right-hand side of both formula, we identify $T_0 = (a_0, -a_1, -a_2, -a_3)$ with the timelike vector $a_0 e_0 - \sum_{i=1}^3 a_i e_i$ at $p$. 
\end{theorem}
To get the same limit using the method of previous sections, we approach $p$ along a spacelike hypersurface. Let $X^0, X^1, X^2, X^3$ be a normal coordinate near $p$. The metric has the expansion
\[ \bar g_{\alpha\beta} = \eta_{\alpha\beta} - \frac{1}{3} \bar R_{\alpha\gamma\beta\delta} X^\gamma X^\delta + \cdots. \] 
Let $M$ be the slice $\{ X^0 = 0 \}$. We consider small spheres $\Sigma_r = \{ X^0=0, (X^1)^2 + (X^2)^2 + (X^3)^2 = r^2 \}$ and balls $B_r = \{ X^0=0, (X^1)^2 + (X^2)^2 + (X^3)^2 \le r^2\}, 0 < r < \epsilon$. The timelike unit normal vector and second fundamental form of $M$ are given by
\[ \vec{n} = \frac{\pl}{\pl X^0} + \frac{1}{3} \bar R_{0jik} X^j X^k \frac{\pl}{\pl X^i} - \frac{1}{6} \bar R_{0j0k} X^j X^k \frac{\pl}{\pl X^0} + O(|X|^3) \]
and
\begin{align*}
k_{ij} &= \frac{1}{2} \lt( \langle D_{\pl_i} \vec{n}, \pl_j + \langle D_{\pl_j} \vec{n}, \pl_i \rangle \rt)\\
&= \frac{1}{2} \pl_0 g_{ij} + \frac{1}{6} \lt( \bar R_{0ijq} + \bar R_{0jiq} \rt)X^q + O(|X|^3)\\
&= -\frac{1}{3} \lt( \bar R_{0iqj} + \bar R_{0jqi}\rt) X^q + O(|X|^2).
\end{align*} 
 We compute on $\Sigma_r$
\begin{align*}
\mbox{tr}_\Sigma k &= O(r^2),\\
\alpha_H &= - k(\pl_a, \nu) + \pl_a \lt( \frac{\mbox{tr}_\Sigma k}{|H|} \rt) + O(r^4) = -\frac{r^2}{6} \bar R_{j0iq} \tilde x^i_a \tilde x^j \tilde x^q + O(r^4),\\
div \alpha_H &= O(r^2).
\end{align*}

We are ready to use Theorem \ref{mass by Jang's equation} to recover the small sphere limits, Theorem \ref{small_sphere_limit}. Consider the nonvacuum case first. By definition, we have $\mu = 8\pi T(e_0, e_0)$ and $J_i = 8\pi T(e_0,e_i)$. Moreover, integrating over $B_r$ provides a factor of $r^3$. Therefore we recover (\ref{CWY_nonvacuum}).

For the non-vacuum case, we again solve the optimal embedding equation and Jang's equation first and then evaluate the three integrals on either $B_r$ or $\Sigma_r$.

\begin{lemma}
The following pair $T_0 = (a_0, -a_1, -a_2, -a_3)$ and
\begin{align*}
\mathcal X^0 &= \frac{a_i}{6a_0} r^3 R_{mn}(p) \tilde x^m \tilde x^n \tilde x^i + O(r^4)\\
\mathcal X^i &= r\tilde x^i - \frac{r^3}{6} R^i_n(p) \tilde x^n - \frac{r^3}{6} R_{mn}(p) \tilde x^m \tilde x^n \tilde x^i+ O(r^4)
\end{align*}
solves the leading order of the optimal embedding equation on $\Sigma_r$. In particular, the above solution gives a time function $\tau = - \mathcal X \cdot T_0$ with
\begin{align*}
\tau = r a_i \tilde x^i - \frac{r^3}{6} a_i R^i_n(p) \tilde x^n + O(r^4).
\end{align*}
\end{lemma}

\begin{lemma}
Let $u$ be the solution of the Dirichlet problem of Jang's equation on $B_r$ with boundary value $\tau$.  Then $u = a_i X^i + \frac{r^2}{6} T_i X^i + \frac{r^2}{6} B_{ijm} X^i X^j X^m - \frac{r^2}{6} a_i R^i_n X^n$, where 
\[ B_{ijm} = -\frac{1}{3} (\delta_{ij} T_m + \delta_{jm} T_i +\delta_{mi} T_j). \]
\end{lemma}
The constants $T_i$ can be solved from Jang's equation. As before, their contribution to each integral would cancel and we do not bother to solve them explicitly here.
\begin{lemma}
\begin{align*}
   &\frac{1}{16\pi} \int_{B_r} \lt[ \lt| k_{ij} - \frac{D_iD_ju}{\sqrt{1 + |Du|^2}} \rt|^2_{\bar g} + 2 |Y|^2_{\bar g} \rt] d\tilde\Omega  \\
=& r^5 \Bigg[ -\frac{1}{54a_0} \sum_m T_m^2 - \frac{1}{27a_0} T_m a_l R^{lm}(p) + \frac{1}{90} \lt( a_0 - \frac{1}{a_0} \rt) R_{ij}(p)R^{ij}(p) - \frac{1}{54a_0} a_i a_j R^{il}(p) R^j_l(p)\\
& + \frac{a_0}{180} \bar W_{0imj}(p) \bar W\indices{_0^{imj}}(p) -\frac{1}{45} a_m \bar W\indices{_0^{imj}}(p) \bar W_{0i0j}(p) \Bigg] + O(r^6).
\end{align*}
\end{lemma}
\begin{proof} We have
\begin{align*}
   & k_{ij} - \frac{D_iD_j u}{\sqrt{1+|Du|^2}} \\
= &  \frac{1}{3} \lt( \frac{1}{a_0} (\delta_{ij} T_m + \delta_{jm} T_i + \delta_{mi} T_j) - \bar W_{0imj} - \bar W_{0jmi} + \frac{a_l}{a_0} \lt( R_{iljm} + R_{imjl} \rt) \rt)X^m + O(r^2)
\end{align*}
Since we are integrating on a ball with radius $r$ instead of $1$, we get an additional factor $r^5$ from
\[ \frac{1}{16\pi} \int_{B_r} X^m X^n dx = \frac{r^5}{60} \delta^{mn}. \] 
\end{proof}
\begin{lemma}
\begin{align*}
\frac{1}{16\pi} \int \lt[ |h_0(s) - h(s)|^2 - (H_0(s) - H(s))^2 \rt] d\tilde\Omega = \frac{r^5}{60a_0} R_{ij}(p)R^{ij}(p) + O(r^6).
\end{align*}
\end{lemma}
\begin{proof}
The argument is almost identical as in the proof of Lemma \ref{lemma_difference_2nd_ff}. We get an additional factor $r^5$ in the last step:
\[ \int_0^r s^4 ds = \frac{r^5}{5}. \]
\end{proof}

\begin{lemma}
\begin{align*}
\frac{1}{8\pi} \int_{\Sigma_r} \lt[ |H|\sqrt{1+|\na\tau|^2} (\cosh\theta' - \cosh\theta) + \Delta\tau (\theta' - \theta)\rt] d\Sigma \\
= \frac{r^5}{54a_0} \lt( \sum_m T_m^2 + 2 T_m a_i R^{im}(p) + a_i a_j R^{il}(p) R^j_l(p) \rt)+ O(r^6).
\end{align*}
\end{lemma}
\begin{proof}
A similar computation as in the proof of Lemma \ref{lemma_gauge} leads to
\begin{align*}
|H|(\sinh\theta - \sinh\theta') = \frac{r}{\sqrt{1+|\na\tau|^2}} \frac{2}{3} \lt( T_m \tilde x^m + a_i R^i_j \tilde x^j \rt) + O(r^2).
\end{align*}
Recall the main term in the integrand is $|H|(\sinh\theta - \sinh\theta')^2$. We get a factor $r^3$ because $|H| = \frac{2}{r} + O(r^{-2})$ and another factor $r^2$ from the area form of $\Sigma_r$.
\end{proof}

Putting the above three lemma together with Theorem \ref{mass by Jang's equation}, we recover (\ref{CWY_vacuum}).

\end{document}